\documentclass[envcountsame,11pt]{article}

\usepackage{graphicx}
\usepackage[hidelinks]{hyperref}
\usepackage{color,soul}

\usepackage{amssymb,amsmath,amsthm,ragged2e}
\usepackage[mathscr]{euscript}

\usepackage{caption,subcaption}
\usepackage{booktabs,multirow}

\usepackage{algorithm}
\usepackage[noend]{algorithmic}
\usepackage{tikz}
\usetikzlibrary{arrows}

\usepackage{cleveref}

\graphicspath{{figs/}}

\newtheorem{theorem}{Theorem} 
\newtheorem{lemma}[theorem]{Lemma} 
\newtheorem{obs}{Observation}

\newcommand{\m}{\xi}

\newcommand{\CG}{\ensuremath{\mathcal{G}}}

\usepackage[letterpaper, margin=1.2in]{geometry}

\Crefname{table}{Table}{Tables}
\Crefname{obs}{Observation}{Observations}
\Crefname{cor}{Corollary}{Corollaries}
\Crefname{algorithm}{Algorithm}{Algorithms}
\Crefname{lemma}{Lemma}{Lemmas}

\newcommand{\ceil}[1]{\left\lceil{#1}\right\rceil}
\newcommand{\floor}[1]{\left\lfloor{#1}\right\rfloor}
\newcommand{\card}[1]{\left|{#1}\right|}

\newcommand{\etal}{{\em et~al.\/}}
\newcommand{\REM}[1]{}
\newcommand{\QED}{}  

\title{New Bounds on the Biplanar and $k$-Planar Crossing Numbers} 

\author{
	Alireza Shavali\thanks{Department of Computer Engineering, 
		Sharif University of Technology, Tehran 14588-89694, Iran. Email: \url{ashavali@ce.sharif.edu}.}
	\and Hamid Zarrabi-Zadeh\thanks{Department of Computer Engineering, 
	Sharif University of Technology, Tehran 14588-89694, Iran. Email: \url{zarrabi@sharif.edu}.}}

\date{}

\begin{document}
\maketitle

\begin{abstract}
The biplanar crossing number of a graph $G$ 
is the minimum number of crossings over all possible drawings of the edges
of $G$ in two disjoint planes. 
We present new bounds on the biplanar crossing number of 
complete graphs and complete bipartite graphs. 
In particular, we prove that 
the biplanar crossing number of complete bipartite graphs
can be approximated to within a factor of $3$,
improving over the best previously known approximation factor of $4.03$.
For complete graphs, 
we provide a new approximation factor of $3.17$,
improving over the best previous factor of $4.34$.
We provide similar improved approximation factors for the $k$-planar 
crossing number of complete graphs and complete bipartite graphs,
for any positive integer $k$.
We also investigate the relation between (ordinary) crossing number 
and biplanar crossing number of general graphs in more depth,
and prove that any graph with a crossing number of at most $10$ is biplanar.
\end{abstract}



\section{Introduction}

An embedding (or drawing) of a graph $G$ in the Euclidean plane 
is a mapping of the vertices of $G$ to distinct points in the plane and 
a mapping of edges to smooth curves between their corresponding vertices. 
A planar embedding is a drawing of the graph such that no two edges cross each other, 
except for possibly in their endpoints. 
A graph that admits such a drawing is called {planar}. 
A \emph{biplanar embedding} of a graph $G=(V,E)$ 
is a decomposition of the graph into two graphs $G_{1}=(V,E_{1})$ and $G_{2}=(V,E_{2})$ 
such that $E=E_{1} \cup E_{2}$ and $E_{1} \cap E_{2} = \emptyset $,
together with planar embeddings of $G_{1}$ and $G_{2}$. 
In this case, we call $G$ biplanar. 
Biplanar embeddings are central to the computation of thickness of graphs~\cite{mutzel1998},
with applications to VLSI design~\cite{owen1971}.
It is well-known that planarity can be recognized in linear time, 
while biplanarity testing is NP-complete~\cite{mas1983}. 

Let $cr(G)$ be the minimum number of edge crossings over all drawings of $G$ in the plane,
and let $cr_{k}(G)$ be the minimum of $cr(G_{1})+cr(G_{2})+ \cdots +cr(G_{k})$ 
over all possible  decompositions of $G$ into $k$ subgraphs $G_{1}, G_{2}, \ldots ,G_{k}$. 
We call $cr(G)$ the \emph{crossing number} of $G$,
and $cr_{k}(G)$ the \emph{$k$-planar crossing number} of $G$. 
Throughout this paper, we only consider \textit{simple drawings} for each subgraph $G_{i}$,
in which 
no two edges cross more than once, 
and no three edges cross at a point
(such drawings are sometimes called nice drawings).
Moreover, we denote by $n$ the number of vertices, and by $m$ the number of edges of a graph.

Determining the crossing number of complete graphs and complete bipartite graphs 
has been the subject of extensive research 
over the past decades. 
In 1955, Zarankiewicz~\cite{zar1955} conjectured that the crossing number $cr(K_{p,q})$ 
of the complete bipartite graph $K_{p,q}$ is equal to 
$$
Z(p,q)  :=  \floor{\frac{p}{2}} \floor{\frac{p-1}{2}} \floor{\frac{q}{2}} \floor{\frac{q-1}{2}}.
$$
He also established a drawing with that many crossings. 
In 1960, Guy~\cite{guy1960} conjectured that the crossing number $cr(K_{n})$
of the complete graph $K_n$ is equal to 
$$
Z(n)  :=  {1 \over 4} \floor{\frac{n}{2}} \floor{\frac{n-1}{2}} \floor{\frac{n-2}{2}} \floor{\frac{n-3}{2}}.
$$
Both conjectures have remained open after more than six decades.
For the biplanar case, even formulating such conjectures seems to be hard. 
As noted in~\cite{cz2006}, techniques like {embedding method} and the {bisection width method} 
which are useful for bounding ordinary crossing numbers 
do not seem applicable to the biplanar case. 

In 1971, Owens~\cite{owen1971} described a biplanar embedding of $K_{n}$
with almost ${7\over 24} Z(n)$ crossings. 
In 2006, Czabarka \etal~\cite{cz2006} presented a biplanar embedding for $K_{p,q}$
with about ${2 \over 9} Z(p, q)$ crossings. 
They also proved that
$cr_{2}(K_{n}) \geq {n^4}/{952}$ and
$cr_{2}(K_{p,q}) \geq {p(p-1)q(q-1)}/{290}$. 
Shahrokhi \etal~\cite{sha2007} generalized these lower bounds 
to the $k$-planar case.
Recently, Pach \etal~\cite{pach2018} proved that for every graph $G$
and any positive integer $k$,
$cr_{k}(G) \leq \left( {2 \over k^2} - {1 \over k^3}\right) cr(G)$. 
This includes as a special case the inequality
$cr_{2}(G) \leq \frac{3}{8}cr(G)$, originally proved by Czabarka \etal~\cite{cz2008}.





\bigskip\noindent{\bf Our results. \ }
In this paper, we present several new bounds for approximating 
the biplanar and $k$-planar crossing number
of complete graphs and complete bipartite graphs.
Given a positive integer $k$ and a real constant $\alpha \geq 1$, 
we say that $cr_k(K_n)$ is \emph{approximated} to within a factor of $\alpha$,
if there is an upper bound $f(n)$ and a lower bound $g(n)$ on the value of $cr_k(K_n)$ 
such that $\lim_{n \rightarrow \infty} {f(n) \over g(n)} \leq \alpha$.
Here, $\alpha$ is called an \emph{asymptotic approximation factor} for $cr_k(K_n)$.
Similarly, 
we say that $cr_k(K_{p,q})$ is approximated to within a factor of $\alpha$, 
if there is an upper bound $f(p,q)$ and a lower bound $g(p,q)$ on the value of $cr_k(K_{p,q})$ 
such that $\lim_{p,q \rightarrow \infty} {f(p,q) \over g(p,q)}$ exists and is no more than $\alpha$.
The results presented in this paper are summarized below.

\begin{itemize}
\item 
	We prove that for all $p,q \geq 21$, 
	$ cr_{2}(K_{p,q}) \geq {p(p-1)q(q-1)}/{216}$.
	This significantly improves the best current lower bound of 
	$cr_{2}(K_{p,q}) \geq {p(p-1)q(q-1)}/{290}$,
	due to Czabarka \etal~\cite{cz2006}.
	Combined with the upper bound of 
	$cr_{2}(K_{p,q}) \leq {2 \over 9} Z(p,q) + o(p^2q^2)$\footnote{By definition,
		$f(x, y) = o(g(x, y))$ if $\lim_{x, y \rightarrow \infty} {f(x, y) \over g(x, y)} = 0$.
	}~\cite{cz2006},
	our result implies an asymptotic approximation factor of 3 for $cr_{2}(K_{p,q}) $,
	improving over the best previously known  approximation factor of $4.03$.

\item 
	For complete graphs, we show that $cr_{2}(K_{n}) \geq \frac{n^4}{694}$,
	improving the  best current lower bound of 
	$cr_{2}(K_{n}) \geq \frac{n^4}{952}$~\cite{cz2006}.
	Combined with the upper bound of $cr_{2}(K_{n}) \leq {7 \over 24} Z(n) + o(n^4)$
	due to Owens~\cite{owen1971},
	we achieve an asymptotic approximation factor of $3.17$ for $cr_{2}(K_{n}) $,
	improving the best previously known approximation factor of $4.34$.
	
\item 
	We investigate the relation between $cr(G)$ and $cr_{2}(G)$  in general graphs,
	and pose a new problem of finding the maximum integer $\m(r)$, 
	for a given integer $r \geq 0$,
	such that $cr(G) \leq \m(r)$ implies $cr_{2}(G) \leq r$, for all graphs $G$.
	For the special case of $r = 0$, we show that $\m(r) \geq 10$.
	It implies that 
	any graph $G$ that can be drawn in the plane with at most 10 crossings is biplanar.

\item 
	We extend our lower bounds for the biplanar crossing number to the $k$-planar case,
	for any positive integer $k$.
	In particular, we show that for sufficiently large $n$,
	$cr_{k}(K_{n}) \geq {n^4}/{(232k^2)}$,
	improving the best current lower bound of $cr_{k}(K_{n}) \geq {n^4}/{(432k^2)}$,
	due to Shahrokhi \etal~\cite{sha2007}.
	Considering the recent upper bound of 
	$cr_{k}(K_{n}) \leq {2 \over k^2} Z(n)$
	proved by Pach \etal~\cite{pach2018},
	we obtain an asymptotic approximation factor of $7.25$ for $cr_{k}(K_{n})$,
	improving the best current approximation factor of $13.5$  available for $cr_{k}(K_{n})$.
	
\item 
	Finally, we  prove that for any positive integer $k$, 
	$cr_{k}(K_{p,q}) \geq {p(p-1)q(q-1)}/{(73.2k^2)}$,
	improving the current lower bound of $cr_{k}(K_{p,q}) \geq {p(p-1)q(q-1)}/{(108k^2)}$
	due to Shahrokhi \etal~\cite{sha2007}.
	Combined with the upper bound of $cr_{k}(K_{n}) \leq {2 \over k^2} Z(p,q)$~\cite{pach2018},
	we obtain an asymptotic approximation factor of $9.15$ for $cr_{k}(K_{p,q}) $,
	improving the best current factor of $13.5$.
	
\end{itemize}

A summary of the asymptotic approximation factors 
for the biplanar and $k$-planar crossing number of $K_n$ and $K_{p,q}$
is presented in \Cref{table:summary}.

\begin{table}[]
\centering
\caption{
Summary of the asymptotic approximation factors 
for the biplanar and $k$-planar crossing number of complete graphs and complete bipartite graphs.
}
\label{table:summary}
\begin{tabular}{ccc}
\hline
\multirow{2}{*}{\sc Crossing Number } & \sc Asymptotic &    \multirow{2}{*}{\sc \hspace{1em} Reference \hspace{1em} }      \\ 
    & \sc\hspace{0em} Approx Factor  \hspace{0em} &    \\ \hline\hline
\multirow{2}{*}{$cr_2(K_{p,q})$} & 4.03                     & ~\cite{cz2006} \\
                                 & \textbf{3}               & {[}This work{]}          \\ \hline
\multirow{2}{*}{$cr_2(K_n)$}     & 4.34                     & \cite{cz2006,owen1971} \\
                                 & \textbf{3.17}            & {[}This work{]}          \\ \hline
\multirow{2}{*}{$cr_k(K_{p,q})$} & 13.5                     & \cite{pach2018,sha2007} \\
                                 & \textbf{9.15}            & {[}This work{]}          \\ \hline
\multirow{2}{*}{$cr_k(K_n)$}     & 13.5                     & \cite{pach2018,sha2007} \\
                                 & \textbf{7.25}            & {[}This work{]}          \\ \hline
\end{tabular}
\end{table}



\section{Preliminaries}

One of the main combinatorial tools typically used for 
deriving lower bounds on the crossing number of graphs
is the counting method (see, e.g.,~\cite{guy1968toroidal,sha1995}).
We use the following generalization of the counting method in this paper.

\begin{lemma}[Counting method]
	\label{lem:counting} 
	Let $G$ be a simple graph that contains $\alpha$ copies of a subgraph $H$.
	If in every $k$-planar drawing of $G$,
	each crossing of the edges belongs to at most $\beta$ copies of $H$, then 
	$$cr_{k}(G) \geq \ceil{{\alpha \over \beta} \, cr_{k}(H)}.$$
\end{lemma}

\begin{proof}
Let $D$ be a $k$-planar drawing of $G$, realizing $cr_{k}(G)$.
For each of the $\alpha$ copies of $H$, $D$ contains a $k$-planar drawing
with at least $cr_{2}(H)$ crossings.
Since each crossing is counted at most $\beta$ times by our assumption, 
the lemma statement follows. 
Note that a ceiling is put in the right-hand side of the inequality, 
because $cr_{k}(G)$ is always an integer.
\end{proof}

\noindent
The following lemma provides another main ingredient used throughout this paper.

\begin{lemma} \label{lem:hereditary} 
	Let $\CG$ be a hereditary class of graphs which is closed 
	under removing edges.
	Let $f$ be a linear function 
	$f(x) = cx$, for some constant $c$,
	and let $g$ be an arbitrary function.
	If for every graph $G$ in $\CG$, $cr(G) \geq f(m) - g(n)$,
	then $cr_k(G) \geq f(m) - k \cdot g(n)$ for all $G \in \CG$
	and all positive integers $k$.
\end{lemma}

\begin{proof}
	Fix a graph $G \in \CG$.
	Let $G = \bigcup_{i=1}^{k} G_k$
	be a decomposition of $G$ 
	into $k$ subgraphs $G_i = (V, E_i)$
	such that $\sum_{i=1}^{k} cr(G_{i})$ is minimum.
	By the hereditary property of $\CG$,
	each $G_i$ is in $\CG$, and hence
	$cr(G_i) \geq f(m_i) - g(n)$, where $m_i = \card{E_i}$.
	Therefore,
	$cr_k(G) = \sum_{i=1}^{k} cr(G_{i}) \geq \sum_{i=1}^{k} (f(m_i) - g(n)) 
	= c\sum_{i=1}^{k} m_i - \sum_{i=1}^{k} g(n)
	= f(m) - k \cdot g(n).$
	\QED
\end{proof}

\section{Lower Bounds for Complete Bipartite Graphs}

In this section, we provide new lower bounds on
the biplanar crossing number of complete bipartite graphs.
In particular, we improve the following bound due to 
Czabarka \etal~\cite{cz2006} which states that for all $p,q \geq 10$,
$$
	cr_{2}(K_{p,q}) \geq \frac{p(p-1)q(q-1)}{290}.
$$

\noindent
From Euler's formula,  
we have $cr(G) \geq m-3(n-2)$ for simple graphs, and
$cr(G) \geq m-2(n-2)$ for bipartite graphs.
Using \Cref{lem:hereditary}, we immediately get a lower bound of
$cr_2(G) \geq m-6(n-2)$ for simple graphs, and a lower bound of
$cr_2(G) \geq m-4(n-2)$ for bipartite graphs.

To establish stronger lower bounds, we need to incorporate more powerful ingredients. 
A graph is called $k$-planar, 
if it can be drawn in the plane in such a way that each edge has at most $k$ crossings.
It is known that every $1$-planar drawing of any $1$-planar graph
has at most $n-2$ crossings~\cite{czap2013}. 
(Note the difference between $k$-planar graphs, and $k$-planar crossing numbers.)
Removing one edge per crossing yields a planar graph.
Therefore, every 1-planar bipartite graph has at most $3n-6$ edges. 
Karpov~\cite{kar2012} proved that for every 1-planar bipartite graph with at least 4 vertices,
the inequality $m \leq 3n-8$ holds. 
In a recent work, Angelini \etal ~\cite{angelini2018} proved that for every 
2-planar bipartite graph we have $m \leq 3.5n-7$.
We use these results to obtain the following stronger lower bound.


\begin{lemma} \label{lem:karlb} 
	For every bipartite graph $G$ with $n \geq 4$, 
	$$cr_{2}(G) \geq 3m-17n+38. $$
\end{lemma}

\begin{proof}
Let $G$ be a bipartite graph with $n$ vertices and $m$ edges. 
Fix a drawing of $G$ with a minimum number of crossings.
If $m > 3.5n - 7$, then by ~\cite{angelini2018}, 
there must be an edge in the drawing with at least three crossings. 
We repeatedly remove such an edge
until we reach a drawing with $\floor{3.5n - 7}$ edges.
Then by Karpov's result there must be an edge in the drawing with at least two crossings. 
Similarly we repeatedly remove such an edge 
until we reach a drawing with $3n-8$ edges.
Let $G'$ be the bipartite graph corresponding to the remaining drawing.
Now, 
\begin{align*}
	 cr(G) 
	 & \geq 3(m - \floor{3.5n - 7}) +2(\floor{3.5n - 7} - (3n-8)) + cr(G') \\
	 & \geq 3(m - \floor{3.5n - 7}) +2(\floor{3.5n - 7} - (3n-8)) + (3n-8) - 2(n-2) \\
	 & \geq 3m - \floor{3.5n - 7} - (3n-8) - 2(n-2) \\
	 & \geq 3m-8.5n+19.
\end{align*}
Applying \Cref{lem:hereditary} yields $cr_2(G) \geq 3m-17n+38$.
\QED
\end{proof}


\noindent
For complete bipartite graphs,
\Cref{lem:karlb} implies that  $cr_{2}(K_{p,q}) \geq 3pq-17(p+q)+38$,
for all $p, q \geq 2$.
We use \Cref{lem:karlb} along with a counting argument
to obtain the following improved bound on $cr_{2}(K_{p,q})$.


\begin{theorem} \label{thm:biplb}
	For all $p,q \geq 21$, 
	$$ cr_{2}(K_{p,q}) \geq \frac{p(p-1)q(q-1)}{216}.$$
\end{theorem}
\begin{proof}

Using the counting method (\Cref{lem:counting}) for $K_{n,n}$ and $K_{n+1,n}$ we have 
$$cr_{2}(K_{n+1,n}) \geq \ceil{\frac{n+1}{n-1}cr_{2}(K_{n,n})}.$$
This is because $K_{n+1,n}$ contains $n+1$ copies of $K_{n,n}$,
and each crossing realized by two edges, belongs to at most 
${n-1 \choose n-2} = n-1$ of these copies.
Using a similar argument for $K_{n+1,n}$ and $K_{n+1,n+1}$, we get 
\begin{equation} \label{eq:recurrence} 
cr_{2}(K_{n+1,n+1}) \geq \ceil{\frac{n+1}{n-1} \ceil{\frac{n+1}{n-1}cr_{2}(K_{n,n})}}.
\end{equation}
By \Cref{lem:karlb}, $cr_{2}(K_{15,15}) \geq 203$.
Plugging into (\ref{eq:recurrence}), yields $cr_{2}(K_{16,16}) \geq 266$. 
Now, we use the recurrence relation (\ref{eq:recurrence}) iteratively from $n=15$ to $21$ to get 
\begin{equation} \label{eq:k28} 
cr_{2}(K_{21,21}) \geq 817.
\end{equation}
We can now apply the counting method 
on $K_{21,21}$ and $K_{p,q}$ to obtain
$$
cr_{2}(K_{p,q}) 
\geq \frac{\binom{p}{21} \binom{q}{21}}{\binom{p-2}{19} \binom{q-2}{19}} \ cr_{2}(K_{21,21})
= \frac{p(p-1)q(q-1)}{21 \times 20 \times 21 \times 20} \ cr_{2}(K_{21,21}).
$$ 
Plugging (\ref{eq:k28}) in the above inequality yields the theorem statement.
\QED
\end{proof}

\REM{
\paragraph{Remark.}
The exact value of the denominator obtained in our proof is around $215.911$.
One may continue applying the recurrence relation (\ref{eq:recurrence}) 
to obtain better bounds for $K_{n,n}$ when $n>21$.
This leads to a slightly improved constant in the denominator,
but it does not seem to reduce the constant below 215. 
Indeed, the denominator seems to converge to a value around $215.131$.
}

\section{Biplanar Crossing Number of  Complete Graphs}
\label{sec:complete}

We now consider the biplanar crossing number of complete graphs. 
Czabarka \etal~\cite{cz2006} used a probabilistic method to prove that for large values of $n$,
$$
	cr_{2}(K_{n}) \geq \frac{n^4}{952}. 
$$
We improve this lower bound using the counting method. 

\begin{theorem} \label{th2}
For all $n \geq 24$,
$$cr_{2}(K_{n}) \geq \frac{n(n-1)(n-2)(n-3)}{698}.$$
\end{theorem}

\begin{proof}
We know from~\cite{ackerman2019} that
for every $G$ with $n \geq 3$,
$cr(G) \geq 5m- \frac{139}{6}(n-2).$
Applying \Cref{lem:hereditary}, we get
$$
	cr_{2}(G) \geq 5m- \frac{139}{3}(n-2).
$$
This in particular implies that $cr_{2}(K_{25}) \geq 435$.
Now, we use the counting method (\Cref{lem:counting}) on $K_{25}$ and $K_{n}$ to get
$$
cr_{2}(K_{n}) 
\geq \frac{\binom{n}{25}cr_{2}(K_{25})}{\binom{n-4}{21}} 
\geq \frac{n(n-1)(n-2)(n-3)}{\frac{25 \times 24 \times 23 \times 22}{435}},
$$
which implies the theorem statement.
\QED
\end{proof}

\noindent
We can slightly improve this result, using an iterative counting method 
similar to what we used in the previous section.

\begin{theorem} \label{th3}
	For large values of $n$,
	$$cr_{2}(K_{n}) \geq \frac{n^4}{694}.$$
\end{theorem}

\begin{proof}
Using the counting method (\Cref{lem:counting}) for $K_{n}$ and $K_{n+1}$ we have,
\begin{equation} \label{eq:rec2}
 cr_2(K_{n+1}) \geq \ceil{\frac{(n+1)cr_{2}(K_{n})}{n-3}}.
\end{equation}

Starting from $cr_{2}(K_{25}) \geq 435$,
we use the recurrence relation (\ref{eq:rec2}) 
iteratively from $n=25$ to $57$ to obtain $cr_{2}(K_{57}) \geq 13667$. 
Now, we use the counting method on $K_{57}$ and $K_{n}$ to get
$$
cr_{2}(K_{n}) 
\geq \frac{\binom{n}{57}cr_{2}(K_{57})}{\binom{n-4}{53}} 
\geq \frac{n(n-1)(n-2)(n-3)}{\frac{57 \times 56 \times 55 \times 54}{13667}}
\geq \frac{n(n-1)(n-2)(n-3)}{693.9},
$$
which implies
$ cr_{2}(K_{n}) \geq \frac{n^4}{694}$
for $n$ sufficiently large.
\QED
\end{proof}


\section{The Maximum Crossing Number that Implies Biplanarity}

Czabarka \etal~\cite{cz2008} defined $c^*$ as the smallest constant such that 
for every graph $G$, $cr_{2}(G) \leq c^{*} \cdot cr(G)$. 
They proved that $0.067 \leq c^{*} \leq \frac{3}{8} = 0.375$.
It is known that $cr(K_{n}) \leq \frac{n^4}{64}$~\cite{white1978}.
By \Cref{th3}, for $n$ sufficiently large, $cr_{2}(K_{n}) \geq \frac{n^4}{694}$. 
Therefore, our results from \Cref{sec:complete} imply 
an improved bound of $c^{*} \geq \frac{64}{694} \approx 0.092$.
In a more general sense, we are interested in the following problem. 

\begin{justify}
	{\bf Problem.}
	Given a positive integer $r$, find the largest integer $\m(r)$
	such that for every graph $G$, $cr(G) \leq \m(r)$ implies $cr_{2}(G) \leq r$.
\end{justify}


\noindent
For the special case of $r=0$, 
the problem is to find the largest integer $\m$
such that drawing a graph with at most $\m$ crossings in the plane 
guarantees that the graph is biplanar.
As proved by Battle \etal~\cite{battle1962} and Tutte~\cite{Tutte1963},
$K_{9}$ is not biplanar.
Moreover, we know that $cr(K_{9})=36$~\cite{lib2001}.
Therefore, $\m(0) < 36$. 

The inequality $cr_{2}(G) \leq \frac{3}{8}cr(G)$, due to Czabarka \etal~\cite{cz2008},
implies that if $cr(G) \leq 2$, then $G$ is biplanar.
Therefore, $\m(0) \geq 2$. 
We can strengthen this bound as follows.
Recall that by Kuratowski's theorem, every nonplanar graph 
contains a subdivision of $K_{3,3}$ or $K_{5}$. 
Therefore, there is no nonplanar graph with less than 9 edges. 
This leads to the following observation.

\begin{obs} \label{obs1}
	Every graph with at most 8 edges is planar. 
	The only nonplanar graph with 9 edges is $K_{3,3}$,
	and the only nonplanar graphs with 10 edges are $K_{5}$, 
	$K_{3,3}$ with an extra edge, and $K_{3,3}$ with a subdivided edge.
\end{obs}
From this simple observation, we can infer that $\m(0) \geq 4$ as follows.
Suppose a graph $G$ is drawn in the plane with at most 4 crossings.
The number of edges involved in these four crossings is at most 8. 
If we remove these 8 edges from the drawing, 
the remaining drawing has no crossing. 
Moreover, the subgraph of $G$ that contains only these 8 (or fewer) edges is planar by Observation~\ref{obs1}. 
Therefore, $G$ is the union of two planar graphs, and hence is biplanar. 
We will significantly improve this lower bound  in the following theorem.

\begin{theorem}
	Every graph $G$ with $cr(G) \leq 10$ is biplanar.
	In other words, $\m(0) \geq 10$. 
\end{theorem}

\begin{proof}
Let $G$ be a graph with $cr(G) \leq 10$.
Fix a drawing of $G$ with a minimum number of crossings.
We repeatedly remove an edge from  the drawing 
that involves in a maximum number of crossings
until there remains no more crossings.
Let $G_{1}$ be the graph corresponding to the remaining drawing,
and $G_{2}$ be the graph formed by the removed edges. 
Clearly, $G_{2}$ has at most 10 edges.
Moreover, $G_{1}$ is planar by construction. 
If $G_{2}$ has 8 or less edges, then it is planar by Observation~\ref{obs1},
and we are done.
Otherwise, $G_{2}$ has 9 or 10 edges.
Note that removing any of these edges from $G$ has removed at least one crossing. 
Therefore, removing any of these edges, except possibly the first one, 
has removed exactly one crossing from $G$. 
By Observation~\ref{obs1}, if $G_{2}$ is not planar, then it is either
$K_{5}$, $K_{3,3}$, $K_{3,3}$ with a subdivided edge, or 
$K_{3,3}$ with an extra edge. 
In the former two cases, let $e$ be the last edge removed from $G$.
Clearly, $e$  was crossing exactly one edge $f$ in $G_1$ just before removal.
Therefore, switching $e$ and $f$ between $G_{1}$ and $G_{2}$ keeps $G_{1}$ planar. 
Moroever, the new $G_{2}$ is planar, 
because it contains no subdivision of $K_{5}$ and $K_{3,3}$.
Hence, $G$ is biplanar in the first two cases.
In the latter two cases, i.e., when
$G_{2}$ is a $K_{3,3}$ with a subdivided edge or a $K_{3,3}$ with an extra edge,
$G_2$ has exactly 10 edges.
Therefore,  removing any of these edges from $G$ has removed exactly one crossing,
which means that any edge in $G$ is crossing at most one  edge.
If $G_{2}$ is a $K_{3,3}$ with a subdivided edge, let $e$ be any edge of $G_{2}$ 
except the two edges forming the subdivided edge,
and if $G_{2}$  is a $K_{3,3}$ with an extra edge,
let $e$ be any edge of $G_{2}$ except this extra edge. 
We know that $e$ was crossing exactly one edge $f$ in $G$. 
Moreover, $f$ was only crossing $e$ in $G$,
and hence, it remains in $G_1$ after removing $e$.
Similar to the previous case, switching $e$ and $f$ between $G_{1}$ and $G_{2}$
completes the proof. 
\QED
\end{proof}

\section{$k$-Planar Crossing Number of $K_{n}$ and $K_{p,q}$}

In this section, we provide improved lower bounds on the $k$-planar crossing number
of complete bipartite and complete  graphs.
Shahrokhi \etal~\cite{sha2007} proved that for any positive integer $k$, 
and sufficiently large integers $p$, $q$, and $n$:  
$$
	cr_{k}(K_{p,q}) \geq \frac{p(p-1)q(q-1)}{108k^2},
$$
and
$$
	cr_{k}(K_{n}) \geq \frac{n(n-1)(n-2)(n-3)}{432k^2}. 
$$

\noindent
We improve these results using the ideas developed in Sections 2 and 3.

\begin{theorem} 
	For all $p,q \geq 8k+2$,
	$$cr_{k}(K_{p,q}) \geq \frac{p(p-1)q(q-1)} {73.2 k^2}.$$
\end{theorem}

\begin{proof}
We apply the counting method (\Cref{lem:counting}) on $K_{8k + 2,8k + 2}$ and $K_{p,q}$. 
As noted in the proof of \Cref{lem:karlb}, for every bipartite graph $G$,
$cr(G) \geq 3m-8.5n+19$.
Therefore, by \Cref{lem:hereditary}, $cr_k(G) \geq 3m - (8.5n - 19) k$.
This yields
$$
	cr_{k}(K_{8k + 2,8k+2}) \geq 56k^2 + 43k +12.
$$
Hence,
\begin{align*}
cr_{k}(K_{p,q}) &
\geq \frac{\binom{p}{8k +2} \binom{q}{8k+2} cr_{k}(K_{8k+2,8k+2})}{\binom{p-2}{8k} \binom{q-2}{8k}}
= \frac{p(p-1)q(q-1)cr_{k}(K_{8k+2,8k+2})}{(8k+2)(8k+1)(8k+2)(8k+1)} \\
& \geq \frac{p(p-1)q(q-1)}{\frac{(8k+2)^2(8k+1)^2}{56k^2+43k+12}} 
\geq \frac{p(p-1)q(q-1)}{\frac{512}{7}k^2},
\end{align*}
which completes the proof.
\QED
\end{proof}

\begin{theorem} 
	For all $n \geq 14k - 3$, 
	$$cr_{k}(K_{n}) \geq \frac{n(n-1)(n-2)(n-3)}{232k^2}.$$
\end{theorem}

\begin{proof}
We use the counting method (\Cref{lem:counting}) for $K_{14k-3}$ and $K_{n}$. 
Recall that for every $G$ with $n \geq 3$,
$cr(G) \geq 5m- \frac{139}{6}(n-2)$~\cite{ackerman2019}.
Therefore, 
$cr_{k}(G) \geq 5m- \frac{139}{6}(n-2)k$ by \Cref{lem:hereditary}.
Thus,
$$
cr_{k}(K_{14k-3}) \geq \frac{497}{3}k^2-\frac{775}{6}k+30.
$$
Therefore,
$$cr_{k}(K_{n}) \geq \frac{\binom{n}{14k-3}cr_{k}(K_{14k-3})}{\binom{n-4}{14k-7}} = \frac{n(n-1)(n-2)(n-3)cr_{k}(K_{14k-3})}{(14k-3)(14k-4)(14k-5)(14k-6)},$$ 
which implies the theorem.
\QED
\end{proof}
 
\section{Conclusion}   
In this paper, we presented several improved bounds on 
the biplanar and $k$-planar crossing number of complete graphs and complete bipartite graphs. 
An obvious open problem is whether the asymptotic approximation factors presented in this paper can be further improved.
We also posed an open problem of finding the largest positive integer $\m(r)$ such that 
$cr(G) \leq \m(r)$ implies $cr_{2}(G) \leq r$. 
In particular, we proved that $10 \leq \m(0) \leq 35$. 
This definition can be easily generalized to the $k$-planar case:
given positive integers $k$ and $r$, find the largest integer $\m_k(r)$ such that 
$cr(G) \leq \m_{k}(r)$ implies $cr_{k}(G) \leq r$. 
Determining the value of $\m_{k}(r)$ is an intriguing problem,
even for the special case of $r=0$.



\bibliography{bibs/full,bibs/ref}

\end{document}